\documentclass[12pt,oribibl]{llncs}
\usepackage{amsfonts, amsmath, amssymb}
\usepackage{enumitem}
\usepackage{multirow}

\usepackage{float}

\usepackage{etex}
\usepackage{pgf}
\usepackage{tikz}
\usepackage[utf8]{inputenc}
\usetikzlibrary{graphs}
\usetikzlibrary{arrows,automata,chains,matrix}
\usetikzlibrary{positioning,scopes}
\usetikzlibrary{decorations.pathreplacing}

\usepackage{fullpage}
\usepackage{hyperref}

\title{Some Complete and Intermediate Polynomials in Algebraic Complexity Theory}

\author{Meena Mahajan \and Nitin Saurabh}

\institute{
The Institute of Mathematical Sciences, Chennai, India \\
\email{meena@imsc.res.in, nitin@imsc.res.in} 
}

\newtheorem{fact}[theorem]{Fact}

\newtheorem{defi}[theorem]{Definition}

\renewcommand{\leq}{\leqslant}
\renewcommand{\geq}{\geqslant}

\newcommand{\nat}{\mathbb{N}}       
\newcommand{\reals}{\mathbb{R}}   
\newcommand{\field}{\mathbb{F}}

\newcommand{\Condition}{$\Mod_p\p \not\subseteq \Ppoly$}
\newcommand{\NCondition}{$\Mod_p\p \subseteq \Ppoly$}
\newcommand{\newton}{\mathsf{Newt}}
\newcommand{\efsize}{\mathsf{c}}
\newcommand{\xc}{\mathsf{xc}}

\newcommand{\tree}{\mathsf{T}}      

\def\vp{\mathsf{VP}}
\def\vnp{\mathsf{VNP}}
\def\vbp{\mathsf{VBP}}
\def\p{\mathsf{P}}
\def\np{\mathsf{NP}}
\def\ph{\mathsf{PH}}
\def\Mod{\mathsf{Mod}}
\def\Perm{\textsc{Permanent}}
\def\Ppoly{\mathsf{P/poly}}
\def\perm{\mathsf{Perm}}
\def\det{\mathsf{Det}}
\def\cut{\mathsf{Cut^q}}
\def\vc{\mathsf{VC^q}}
\def\cis{\mathsf{CIS^q}}
\def\3dm{\mathsf{3DM^q}}
\def\x3c{\mathsf{X3C^q}}
\def\sat{\mathsf{Sat^q}}

\def\ham{\mathsf{Clow^q}}

\def\hom{{\bf Hom}}
\def\mon{{\emph{mon}}}

\begin{document}
\maketitle

\begin{abstract}
We provide a list of new natural $\vnp$-intermediate polynomial
families, based on basic (combinatorial) $\np$-complete problems that
are complete under \emph{parsimonious} reductions.  Over finite
fields, these families are in $\vnp$, and under the plausible
hypothesis \Condition, are neither $\vnp$-hard (even under
oracle-circuit reductions) nor in $\vp$. Prior to this, only the Cut
Enumerator polynomial was known to be $\vnp$-intermediate, as shown by
B\"{u}rgisser in 2000.

We next show that over rationals and reals, two of our intermediate
polynomials, based on satisfiability and Hamiltonian cycle, are not
monotone affine polynomial-size projections of the permanent. This
augments recent results along this line due to Grochow.

Finally, we describe a (somewhat natural) polynomial defined
independent of a computation model, and show that it is $\vp$-complete
under polynomial-size projections. This complements a recent result of
Durand et al.\ (2014) which established $\vp$-completeness of a
related polynomial but under constant-depth oracle circuit
reductions. Both polynomials are based on graph homomorphisms. A
simple restriction yields a family similarly complete for $\vbp$.
\end{abstract}

\section{Introduction}
\label{sec:intro}
The algebraic analogue of the $\p$ versus $\np$ problem, famously referred 
to as the $\vp$ versus $\vnp$ question, is one of the most significant 
problem in algebraic complexity theory. Valiant~\cite{Valiant-stoc79a} showed 
that the $\Perm$ polynomial is $\vnp$-complete (over fields of char $\neq$ 2). 
A striking aspect of this polynomial is that the underlying decision problem, 
in fact even the search problem, is in $\p$. Given a graph, we can decide in 
polynomial time whether it has a perfect matching, and if so find a maximum 
matching in polynomial time~\cite{Edmonds65}. Since the underlying problem 
is an easier problem, it helped in establishing $\vnp$-completeness of 
a host of other polynomials by a reduction from the $\Perm$ polynomial 
(cf.~\cite{Burgisser-book00}). Inspired from classical results in structural 
complexity theory, in particular~\cite{Ladner75}, 
B\"{u}rgisser~\cite{Burgisser99-j} proved that if 
Valiant's hypothesis (i.e. $\vp \neq \vnp$) is true, then, over any field 
there is a $p$-family in $\vnp$ which is neither in $\vp$ nor $\vnp$-complete 
with respect to $c$-reductions. Let us call such polynomial families  
$\vnp$-intermediate (i.e. in $\vnp$, not $\vnp$-complete, not in $\vp$). 
Further, B\"{u}rgisser~\cite{Burgisser99-j} showed that  over finite
fields, a \emph{specific} family of polynomials is $\vnp$-intermediate,
provided the polynomial hierarchy $\ph$ does not collapse 
to the second level. On an intuitive level these polynomials enumerate 
\emph{cuts} in a graph. This is a remarkable result, when compared with the 
classical $\p$-$\np$ setting or the BSS-model. 
Though the existence of problems 
with intermediate complexity has been established in the latter settings, 
due to the involved ``diagonalization'' arguments used to construct them, 
these problems seem highly unnatural. That is, 
their definitions are not motivated by an underlying combinatorial 
problem but guided by the needs of the proof and, hence, seem 
artificial. The question of 
whether there are other naturally-defined $\vnp$-intermediate 
polynomials was left open by B\"{u}rgisser~\cite{Burgisser-book00}. We remark 
that to date the \emph{cut enumerator} polynomial from \cite{Burgisser99-j} 
is the only known example of a natural polynomial family 
that is $\vnp$-intermediate. 

The question of whether the classes $\vp$ and $\vnp$ are distinct is often 
phrased as whether $\perm_n$ is \emph{not} a quasi-polynomial-size 
projection of $\det_n$. 
The importance of this reformulation stems from the fact that it is a purely 
algebraic statement, devoid of any dependence on circuits. While we have 
made very little progress on this question of determinantal complexity 
of the permanent, 
the progress in restricted settings has been considerable. One of the success 
stories in theoretical computer science is unconditional lower bound 
against monotone computations~\cite{Razborov85-c,Razborov85-p,AB87}. 
In particular, 
Razborov~\cite{Razborov85-p} proved that 
computing the permanent over the Boolean semiring requires monotone circuits of 
size at least $n^{\Omega(\log n)}$. Jukna~\cite{Jukna14} observed that if 
the Hamilton cycle polynomial is a monotone $p$-projection of the permanent, 
then, since the clique polynomial is a monotone projection 
of the Hamiltonian cycle~\cite{Valiant-stoc79a} and the clique requires 
monotone circuits of exponential size~\cite{AB87}, one would get a lower 
bound of $2^{n^{\Omega(1)}}$ for monotone circuits computing the permanent, thus 
improving on~\cite{Razborov85-p}. The importance of this observation is also 
highlighted by the fact that such a monotone $p$-projection, over the reals, 
would give an alternate proof of the result of Jerrum and Snir
\cite{JS82}  that computing the permanent 
by monotone circuits over $\reals$ requires size at least $2^{n^{\Omega(1)}}$. 
(Jerrum and Snir~\cite{JS82} proved that the permanent requires monotone 
circuits of size $2^{\Omega(n)}$ over $\reals$ and the tropical
semiring.) The first 
progress on this question raised in~\cite{Jukna14} was made recently by 
Grochow~\cite{Grochow15}. He showed that the Hamiltonian cycle polynomial 
is not a monotone sub-exponential-size projection of the permanent. This 
already answered Jukna's question in its entirety, but Grochow~\cite{Grochow15} 
used his techniques to further establish that polynomials like 
the perfect matching polynomial, and even the $\vnp$-intermediate 
cut enumerator polynomial of 
B\"{u}rgisser~\cite{Burgisser99-j}, are not monotone polynomial-size 
projections of the permanent. This raises an intriguing question of whether 
there are other such non-negative polynomials which share this property. 

While the $\perm$ vs $\det$ problem has become synonymous with the 
$\vp$ vs $\vnp$ question, there is a somewhat unsatisfactory feeling about it. This 
rises from two facts: one, that the $\vp$-hardness of the determinant is known 
only under the more powerful quasi-polynomial-size projections, and, 
second, the lack of natural $\vp$-complete polynomials (with respect
to polynomial-size projections) in the literature. (In fact, with
respect to $p$-projections, the determinant is complete for the
possibly smaller class $\vbp$ of polynomial-sized algebraic branching
programs.) 
To remedy this situation, it seems crucial to understand the computation 
in $\vp$. B\"{u}rgisser~\cite{Burgisser-book00} showed that a generic 
polynomial family constructed using a topological sort of a generic $\vp$ 
circuit, while controlling the degree, is complete for $\vp$. 
Raz~\cite{Raz10}, using the depth reduction of~\cite{VSBR83}, 
showed that a family of ``universal circuits'' is $\vp$-complete. 
Thus both families directly depend on the circuit definition or 
characterization of $\vp$. Last year, Durand et al.~\cite{DMMRS14} 
made significant progress and provided a natural, first of its kind, 
$\vp$-complete polynomial. However, the natural polynomials studied by 
Durand et al.\ lacked a bit of punch because their completeness 
was established under polynomial-size \emph{constant depth
  c-reductions} rather than  projections. 

In this paper, we make progress on all three fronts. First, we provide 
a list of new natural polynomial families,
based on basic (combinatorial) $\np$-complete problems~\cite{GJ79}
whose completeness is via
\emph{parsimonious} reductions~\cite{Simon77}, that are
$\vnp$-intermediate over finite fields (Theorem~\ref{thm:intermediate}). 
Then, we show that over reals, some of our intermediate polynomials 
are not monotone affine polynomial-size projections of the permanent
(Theorem~\ref{thm:mon-proj-lb}).  
As in~\cite{Grochow15}, the lower bound results about monotone affine
projections  
are unconditional. Finally, we improve upon \cite{DMMRS14} 
by characterizing $\vp$ and establishing a natural $\vp$-complete polynomial 
under polynomial-size projections (Theorem~\ref{thm:lb-vp}).  A
modification yields a family similarly complete for $\vbp$
(Theorems~\ref{thm:lb-vbp}, \ref{thm:lb-vbp-char-not-2}).


\paragraph{Organization of the paper.} We give basic definitions in 
Section~\ref{sec:prelim}. Section~\ref{sec:intermediate} contains our
discussion on intermediate polynomials. In Section~\ref{sec:monotone} 
we establish lower bounds under monotone affine projections. 
The discussion on completeness results 
appears in Section~\ref{sec:completeness}.   We end in
Section~\ref{sec:concl} with some interesting questions for further
exploration.

\section{Preliminaries}
\label{sec:prelim}
\subsubsection*{Algebraic complexity:}
We say that a polynomial $f$ is a \emph{projection} of $g$ if $f$ can be 
obtained from $g$ by setting the variables of $g$ to either constants 
in the field, or to the variables of $f$. 
A sequence $(f_n)$ is a $p$-\emph{projection} of $(g_m)$, 
if each $f_n$ is a projection of $g_t$ for some
$t=t(n)$ polynomially bounded in $n$. There are other notions of reductions 
between families of polynomials, like
\emph{c-reductions} (polynomial-size oracle circuit reductions),
\emph{constant-depth c-reductions}, and \emph{linear p-projections}. 
For more on these reductions, see~\cite{Burgisser-book00}. 

An arithmetic circuit is a directed acyclic graph with leaves labeled 
by variables or constants from an underlying field, internal nodes 
labeled by field operations $+$ and $\times$, 
and a designated output gate. Each node computes a 
polynomial in a natural way. The polynomial computed by a circuit 
is the polynomial computed at its output gate. A \emph{parse tree} 
of a circuit captures monomial generation within the
circuit. Duplicating gates as needed, unwind
the circuit into a formula (fan-out one); a parse tree is a minimal
sub-tree (of this unwound formula) that contains the output gate, that
contains all children of  each included $\times$ gate, and that
contains exactly one child of each included $+$ gate. 
For a complete definition see~\cite{MP08}. A circuit is said to be
\emph{skew} if at every $\times$ gate, at most one incoming edge is
the output of another gate. 

A family of polynomials $(f_n(x_1,\ldots ,x_{m(n)}))$ is called 
a $p$-family if both the degree $d(n)$ of $f_n$ and the number of
variables $m(n)$ are bounded by a 
polynomial in $n$.
A $p$-family is in $\vp$ (resp.\ $\vbp$)
if a  circuit family (skew circuit family, resp.) $(C_n)$ of size 
polynomially bounded in $n$ computes it. A sequence of polynomials 
$(f_n)$ is in $\vnp$ if there exist a sequence $(g_n)$ in $\vp$, 
and polynomials $m$ and $t$ such that for all $n$, 
\(f_n(\bar{x}) = \sum_{\bar{y}\in\{0,1\}^{t(\bar{x})}} 
g_n(x_1,\ldots,x_{m(n)},y_1,\ldots ,y_{t(n)}).\)
($\vbp$ denotes the algebraic  analogue of branching programs. Since
these are equivalent to skew circuits, we directly use a skew circuit
definition of $\vbp$.)

\subsubsection*{Boolean complexity:}
We need some basics from Boolean complexity theory. 
Let $\Ppoly$ denote the class of languages decidable by polynomial-sized 
Boolean circuit families. A function $\phi:\{0,1\}^\ast \to \nat$ 
is in \#$\p$ if there exists a polynomial $p$ and a polynomial time 
deterministic Turing machine $M$ such that for all $x \in \{0,1\}^\ast$, 
$f(x) = |\{y \in \{0,1\}^{p(|x|)} \mid M(x,y)=1\}|$. 
For a prime $p$, define 
\begin{align*}
\#_p\p &= \{\psi:\{0,1\}^\ast \to \field_p \mid \psi(x) = \phi(x)
\bmod p \textrm{ for some $\phi \in \#\p$} \}, \\
\Mod_p\p &= \{L \subseteq \{0,1\}^{\ast} \mid \textrm{ for some $\phi
  \in \#\p$, } x \in L \iff \phi(x) \equiv 1 \bmod p\} 
\end{align*}
It is easy to see that if $\phi:\{0,1\}^\ast \to \nat$ is \#$\p$-complete with 
respect to parsimonious reductions (that is, for every $\psi \in \#P$,
there is a polynomial-time computable function $f:\{0,1\}^*
\rightarrow \{0,1\}^*$ such that for all $x\in \{0,1\}^*$,
$\psi(x) = \phi(f(x))$), then the language 
$L=\{x\mid \phi(x)\equiv 1 \bmod p\}$ is $\Mod_p\p$-complete with respect to 
 many-one reductions. 

 \subsubsection*{Graph Theory:}
We consider the treewidth and pathwidth parameters for an undirected
graph. 
We will work with a ``canonical'' form of decompositions which is 
generally useful in dynamic-programming algorithms. 

A \emph{(nice) tree decomposition} of a graph $G$ 
is a pair $\mathcal{T} = (T, \{X_t\}_{t\in V(T)})$, where $T$ is 
a tree, rooted at $X_r$, whose 
every node $t$ is assigned a vertex subset $X_t \subseteq V(G)$, called a 
bag, such that the following conditions hold:
\begin{enumerate}
\item $X_r = \emptyset$, $|X_\ell| = 1$ for every leaf $\ell$ of $T$, 
  and $\cup_{t\in V(T)} X_t = V(G)$. \\
  That is, the root contain the empty bag, 
the leaves contain singleton sets, and every vertex of $G$ is in at least 
one bag.
\item For every $(u,v) \in E(G)$, there exists a node $t$ of $T$ such
  that $\{u,v\} \subseteq X_t$.
\item For every $u \in V(G)$, the set $T_u = \{t \in V(T) \mid u \in X_t\}$ 
induces a connected subtree of $T$. 
\item Every non-leaf node $t$ of $T$ is of one of the following three types: 
  \begin{itemize}
    \item \textbf{Introduce node:}  $t$ has exactly once child $t'$, and 
      $X_t = X_{t'} \cup \{v\}$ for some vertex $v \notin X_{t'}$. 
      We say that $v$ is \emph{introduced} at $t$.
    \item \textbf{Forget node:}  $t$ has exactly one child $t'$, and
      $X_t = X_{t'} \setminus \{w\}$ for some vertex $w \in X_{t'}$. We say 
      that $w$ is \emph{forgotten} at $t$.
    \item \textbf{Join node:}  $t$ has two children $t_1,t_2$, and
      \(X_t = X_{t_1} = X_{t_2}.\)
  \end{itemize}
\end{enumerate}
The \emph{width} of a tree decomposition $\mathcal{T}$
is one less than the size of the largest bag;
that is,  $\max_{t\in V(T)}|X_t| - 1$.
The \emph{tree-width} of a graph $G$ is  
the minimum possible width of a tree decomposition of $G$. 

In a similar way we can also define a \emph{nice path decomposition} 
of a graph. For a complete definition we refer to~\cite{fpt-book15}.

A sequence $(G_n)$ of graphs is called a $p$-family if
the number of vertices in $G_n$ is polynomially bounded in $n$. It 
is further said to have \emph{bounded} tree(path)-width if for some
absolute constant $c$ independent of $n$, 
the tree(path)-width of each graph in the sequence is bounded by $c$.

A \emph{homomorphism} from $G$ to $H$ is a map from $V(G)$ to $V(H)$ 
preserving edges. A graph is called \emph{rigid}
if it has \emph{no} homomorphism to itself other than the identity map. 
Two graphs $G$ and $H$ are 
called \emph{incomparable} if there are \emph{no} homomorphisms from $G \to H$ 
as well as $H \to G$. It is known that asymptotically 
almost all  graphs are rigid, and almost all pairs of nonisomorphic
graphs are also incomparable.  
For the purposes of this paper, we only need a collection of three 
rigid and mutually incomparable graphs. For more details, 
we refer to~\cite{hn-book04}. 

\section{$\vnp$-intermediate}
\label{sec:intermediate}


In \cite{Burgisser99-j}, B\"{u}rgisser showed that unless PH collapses
to the second level, an explicit family of polynomials, called the cut
enumerator polynomial, is $\vnp$-intermediate. He raised the question,
recently highlighted again in \cite{Grochow15}, of whether there are other such
natural $\vnp$-intermediate polynomials. In this section we show that
in fact his proof strategy itself can be adapted
to other polynomial families as well. The strategy can be described
abstractly as follows: Find an explicit polynomial family $h=(h_n)$
satisfying the following properties.
\begin{description}
\item[M: Membership.] The family is in $\vnp$. 
\item[E: Ease.] Over a field $\field_q$ of size $q$ and characteristic
  $p$, $h$  can be evaluated in $\p$.  Thus if
   $h$ is $\vnp$-hard, then we can efficiently compute \#$\p$-hard
  functions, modulo $p$. 
\item[H: Hardness.] The monomials of $h$ encode solutions to a problem that is
  \#$\p$-hard via parsimonious reductions.  Thus if $h$ is in $\vp$, then
  the number of solutions, modulo $p$, can be extracted using coefficient
  computation.
\end{description} 
Then, unless \NCondition\ (which in turn implies that PH collapses to
the second level, \cite{KL82}), $h$ is $\vnp$-intermediate.

We provide a list of $p$-families that, under the same condition
\Condition, are
$\vnp$-intermediate.  All these polynomials are based on basic
combinatorial $\np$-complete problems that are complete under
parsimonious reduction.

\noindent
(1)~The \emph{satisfiablity} polynomial $\sat = (\sat_n)$: For each $n$,
let $\mathsf{Cl}_n$ denote the set of all possible clauses of size 3
over $2n$ literals.  There are $n$ variables $\tilde{X} =
\{X_i\}_{i=1}^n$, and also $8n^3$ clause-variables $\tilde{Y} = \{Y_c\}_{c
  \in \mathsf{Cl}_n}$, one for each 3-clause $c$.
\[\sat_n := \sum_{a \in \{0,1\}^n} \left(\prod_{i\in[n]: a_i=1} X_i^{q-1}\right)
\left(\prod_{\substack{c~\in \mathsf{Cl}_n \\ ~a\textrm{ satisfies }c}} Y_c^{q-1} \right).\]

For the next three polynomials, we consider the complete graph $G_n$
on $n$ nodes, and we have the set of variables $\tilde{X}=\{X_e\}_{e \in
  E_n}$ and $\tilde{Y}=\{Y_v\}_{v \in V_n}$. 

\noindent
(2)~The \emph{vertex cover} polynomial $\vc = (\vc_n)$:
\[\vc_n := \sum_{S\subseteq V_n} \left(\prod_{e \in E_n \colon e \textrm{ is incident 
on }S}  X_e^{q-1}\right) \left(\prod_{v \in S} Y_v^{q-1} \right).\] 

\noindent
(3)~The \emph{clique/independent set} polynomial $\cis = (\cis_n)$:
\[\cis_n := \sum_{T\subseteq E_n} \left(\prod_{e \in T}  X_e^{q-1}\right)
\left(\prod_{v\textrm{ incident on }T} Y_v^{q-1} \right).\] 

\noindent
(4)~The \emph{clow} polynomial $\ham = (\ham_n)$: A clow in an $n$-vertex
graph is a closed walk of length exactly $n$, in which the minimum
numbered vertex (called the head) appears exactly once. 
\[\ham_n := \sum_{w:\textrm{ clow of length }n}
\left(\prod_{e:\textrm{ edges in }w}  X_e^{q-1}\right)
\left(\prod_{\substack{v:\textrm{ vertices in }w\\ \textrm{(counted only once)}}} Y_v^{q-1} \right).\] 
If an edge $e$ is used $k$ times in a clow, 
it contributes $X_e^{k(q-1)}$ to the monomial.
But a vertex $v$
contributes only $Y_v^{q-1}$ even if it appears more than once.
More precisely, 
\[\ham_n := \sum_{\substack{w = \langle v_0, v_1, \ldots , v_{n-1}
    \rangle: \\
        \forall j > 0, ~~v_0 < v_j
}}
\left(\prod_{i\in [n]}  X_{(v_{i-1},v_{i\bmod n})}^{q-1}\right)
\left(\prod_{v \in \{v_0, v_1, \ldots , v_{n-1} \}}
Y_v^{q-1} \right).\] 

\noindent
(5)~The \emph{3D-matching} polynomial $\3dm = (\3dm_n)$:
Consider the complete tripartite hyper-graph, where each part in the
partition $(A_n, B_n, C_n)$ 
contain $n$ nodes, and each hyperedge has exactly one node from each
part.   We have variables $X_e$ for hyperedge $e$ and
$Y_v$ for node $v$. 
\[\3dm_n := \sum_{M \subseteq A_n \times B_n \times C_n}
\left(\prod_{e \in M}  X_e^{q-1}\right)
\left(\prod_{\substack{v \in M\\ \textrm{(counted only once)}}} Y_v^{q-1} \right).\]

We show that if \Condition, then all five polynomials defined above
are $\vnp$-intermediate.
\begin{theorem}
\label{thm:intermediate}
Over a finite field $\field_q$ of characteristic $p$, the polynomial
families $\sat$, $\vc$, $\cis$, $\ham$, and $\3dm$,  are in
$\vnp$.  Further, if \Condition, then they are all $\vnp$-intermediate; that
is, neither in $\vp$ nor $\vnp$-hard with respect to $c$-reductions.
\end{theorem}
\begin{proof}
(M) An easy way to see membership in $\vnp$ is to
use Valiant's criterion (\cite{Valiant-stoc79a}; see also Proposition
2.20 in \cite{Burgisser-book00}); the 
coefficient of any  monomial can be computed efficiently, hence
the polynomial is in $\vnp$. This establishes membership for all 
families. 

We first illustrate the rest of 
the proof by showing that the polynomial $\sat$
satisfies the properties (H), (E). 

(H): Assume $(\sat_n)$ is in $\vp$, via polynomial-sized circuit family
$\{C_n\}_{n \ge 1}$.  We will use $C_n$ to give a $\Ppoly$ upper bound
for computing the number of satisfying assignments of a 3-CNF formula,
modulo $p$.  Since this question is complete for $\Mod_p\p$, the upper
bound implies $\Mod_p\p$ is in $\Ppoly$.

Given an instance  $\phi$ of 3SAT, with $n$ variables and $m$ clauses,
consider the projection of $\sat_n$ obtained by setting 
all $Y_c$ for $c \in \phi$ to $t$, and all other variables to 1. This
gives the polynomial 
$\sat\phi(t) = \sum_{j=1}^m d_j t^{j(q-1)}$  where
$d_j$ is the number of assignments (modulo $p$) 
that satisfy exactly  $j$ clauses in $\phi$. Our goal is to compute $d_m$.

We convert the circuit $C$ into a circuit $D$ that compute
elements of $\field_q[t]$ by explicitly giving their coefficient
vectors, so that we can pull out the desired coefficient. (Note that
after the projection described above, $C$ works over the polynomial
ring $\field_q[t]$.) Since the polynomial computed by $C$ is of degree
$m(q-1)$, we need to compute the coefficients of all intermediate
polynomials too only upto degree $m(q-1)$. Replacing $+$ by gates
performing coordinate-wise addition, $\times$ by a sub-circuit
performing (truncated) convolution, and supplying appropriate
coefficient vectors at the leaves gives the desired circuit.
Since the number of clauses, $m$, is polynomial in $n$, the circuit
$D$ is also of polynomial size.
Given the description of $C$ as advice, the circuit $D$ can be  
evaluated in $\p$, giving a $\Ppoly$ algorithm for computing   
\#3-SAT($\phi$) $\bmod~p$.  Hence \NCondition.

(E) Consider an assignment to $\tilde{X}$ and $\tilde{Y}$ variables in
$\field_q$.  Since all exponents are multiples of $(q-1)$, it suffices
to consider $0/1$ assignments to $\tilde{X}$ and $\tilde{Y}$. Each
assignment $a$ contributes 0 or 1 to the final value; call it a
contributing assignment if it contributes 1. So we just need to
count the number of contributing assignments. An assignment 
$a$ is contributing
exactly when $\forall i\in[n]$, $X_i = 0 \Longrightarrow a_i = 0$,
and $\forall c \in \mathsf{Cl}_n$, $ Y_c=0 \Longrightarrow a \textrm{
  does not satisfy } c $.
These two conditions, together with the values of the $X$ and $Y$
variables, constrain
many bits of a contributing assignment; an inspection
reveals how many (and which) bits are so constrained.
If any bit is constrained in
conflicting ways (for example, $X_i = 0$, and $Y_c = 0$ for some
clause $c$ containing the literal $\bar{x}_i$), then no assignment
is contributing (either $a_i=1$  and  the $X$ part becomes
zero due to $X_i^{a_i}$, or $a_i=0$ and the $Y$ part becomes zero
due to $Y_c$). 
Otherwise, some bits of a potentially contributing assignment are
constrained by $X$ and $Y$, and the remaining bits can be set in any
way. Hence the total sum is 
precisely \(2^{(\textrm{\# unconstrained bits})} \bmod~p\).

Now assume $\sat$ is $\vnp$-hard.
Let $L$ be any language in $\Mod_p\p$, witnessed via \#$\p$-function
$f$. (That is,  $x\in L \Longleftrightarrow f(x)
\equiv 1 \bmod p$.)  By the results of \cite{Burg00, Burgisser-book00},
there exists a $p$-family $r=(r_n) \in \vnp_{\field_p}$ such that
\(\forall n,~\forall x \in \{0,1\}^n,~r_n(x) = f(x) \bmod p.\) By
assumption, there is a $c$-reduction from $r$ to $\sat$.
We use the oracle circuits from this reduction to decide instances of
$L$. On input $x$, the advice is the circuit $C$ of appropriate size
reducing $r$ to $\sat$.  We evaluate this circuit bottom-up. At the
leaves, the values are known. At $+$ and $\times$ gates, we perform
these operations in $\field_q$. At an oracle gate, the paragraph above
tells us how to evaluate the gate.   So the circuit
can be evaluated in polynomial time, showing that $L$ is in
$\Ppoly$. Thus \NCondition.

For the other four families, it suffices to show the following, 
since the rest is identical as for $\sat$.
\begin{description}
\item[H'.] The monomials of $h$ encode solutions to a problem that is
  \#$\p$-hard via parsimonious reductions.
\item[E'.] Over $\field_q$, $h$  can be evaluated in $\p$.  
\end{description} 
We describe this for the polynomial families one by one.

\subsubsection*{The \emph{vertex cover} polynomial $\vc = (\vc_n)$:}
\[\vc_n := \sum_{S\subseteq V_n} \left(\prod_{e \in E_n \colon e \textrm{ is incident 
on }S}  X_e^{q-1}\right) \left(\prod_{v \in S} Y_v^{q-1} \right).\] 

\noindent
(H'): Given an instance of vertex cover $A = (V(A),E(A))$ such that
$|V(A)| = n$ and $|E(A)| = m$, we show how $\vc_n$ encodes the number
of solutions of instance $A$. Consider the following projection of
$\vc_n$.  Set $Y_v = t$, for $v \in V(A)$. For $e \in E(A),$ set $X_e
= z$; otherwise $e \notin E(A)$ and set $X_e = 1$. Thus, we have
\[\vc_n(z,t) = \sum_{S\subseteq V_n}  z^{(\textrm{\# edges incident on
  }S)(q-1)} t^{|S|(q-1)}.\]  
Hence, it follows that the number of vertex cover of size $k$, modulo
$p$, is the coefficient of $z^{m(q-1)}t^{k(q-1)}$ in $\vc_n(z,t)$.

\noindent
(E'): 
Consider the weighted graph given by the values of $\tilde{X}$ and
$\tilde{Y}$ variables.  Each subset $S\subseteq V_n$ contributes $0$
or $1$ to the total.  A subset $S\subseteq V_n$ contributes $1$ to
$\vc_n$ if and only if every vertex in $S$ has non-zero weight, and every edge
incident on each vertex in $S$ has non-zero weight. That is, $S$ is a
subset of full-degree vertices. Therefore, the total sum is
$2^{(\textrm{\# full-degree vertices})} \bmod p$.

\subsubsection*{The \emph{clique/independent set} polynomial $\cis = (\cis_n)$:}
\[\cis_n := \sum_{T\subseteq E_n} \left(\prod_{e \in T}  X_e^{q-1}\right)
\left(\prod_{v\textrm{ incident on }T} Y_v^{q-1} \right).\] 

\noindent
(H'): Given an instance of clique $A = (V(A),E(A))$ such that $|V(A)|
= n$ and $|E(A)| = m$, we show how $\cis_n$ encodes the number of
solutions of instance $A$. Consider the following projection of
$\cis_n$.  Set $Y_v = t$, for $v \in V(A)$. For $e \in E(A),$ set $X_e
= z$; otherwise $e \notin E(A)$ and set $X_e = 1$. (This is the same
projection as used for vertex cover.) Thus, we have
\[\cis_n(z,t) = \sum_{T\subseteq E_n}  z^{|T\cap E(A)|(q-1)} t^{(\textrm{\# vertices incident on }T)(q-1)}.\] 
Now it follows easily that the number of cliques of size $k$, modulo
$p$, is the coefficient of $z^{{k \choose 2}(q-1)} t^{k(q-1)}$ in
$\cis_n(z,t)$.

\noindent
(E'): Consider the weighted graph given by the values of $\tilde{X}$
and $\tilde{Y}$ variables.  Each subset $T\subseteq E_n$ contributes
$0$ or $1$ to the sum.  A subset $T\subseteq E_n$ contributes $1$ to
the sum if and only if all edges in $T$ have non-zero weight, and every vertex
incident on $T$ must have non-zero weight. Therefore, we consider the
graph induced on vertices with non-zero weights. Any subset of edges
in this induced graph contributes $1$ to the total sum; all other
subsets contribute 0. Let $\ell$ be
the number of edges in the induced graph with non-zero weights. Thus,
the total sum is $2^{\ell} \bmod p$.

\subsubsection*{The \emph{clow} polynomial $\ham = (\ham_n)$:}
A clow in an $n$-vertex
graph is a closed walk of length exactly $n$, in which the minimum
numbered vertex (called the head) appears exactly once. 
\[\ham_n := \sum_{w:\textrm{ clow of length }n}
\left(\prod_{e:\textrm{ edges in }w} X_e^{q-1}\right)
\left(\prod_{\substack{v:\textrm{ vertices in }w\\ \textrm{(counted
      only once)}}} Y_v^{q-1} \right).\] (If an edge $e$ is used $k$
times in a clow, it contributes $X_e^{k(q-1)}$ to the monomial.)

\noindent
(H'): Given an instance $A = (V(A),E(A))$ of the Hamiltonian cycle problem 
with $|V(A)| = n$ and $|E(A)| = m$, we show how $\ham_n$ encodes
the number of Hamiltonian cycles in $A$. Consider the following
projection of $\ham_n$.  Set $Y_v = t$, for $v \in V(A)$. For $e \in
E(A),$ set $X_e = z$; otherwise $e \notin E(A)$ and set $X_e =
1$. (The same projection was used for $\vc$ and $\cis$.) 
Thus, we have
\[\ham_n (z,t) = \sum_{w:\textrm{ clow of length }n}
\left(\prod_{e:\textrm{ edges in }w\cap E(A)} z^{q-1}\right)
\left(\prod_{\substack{v:\textrm{ vertices in }w\\ \textrm{(counted
      only once)}}} t^{q-1} \right).\] From the definition, it now
follows that number of Hamiltonian cycles in $A$, modulo $p$, is the
coefficient of $z^{n(q-1)}t^{n(q-1)}$.

\noindent
(E'): To evaluate $\ham_n$ on instantiations of $\tilde{X}$ and
$\tilde{Y}$ variables, we consider the weighted graph given by the
values to the variables.  We modify the edge weights as follows: if an
edge is incident on a node with zero weight, we make its weight $0$
irrespective of the value of the corresponding $X$ variable.  Thus,
all zero weight vertices are isolated in the modified graph $G$.
Hence, the total sum is equal to the number of closed walks of length
$n$, modulo $p$, in this modified graph. This can be computed in
polynomial time using matrix powering as follows: Let $G_i$ denote the
induced subgraph of $G$ with vertices $\{ i, \ldots ,n\}$, and let
$A_i$ be its adjacency matrix. We represent $A_i$ as an $n\times n$
matrix with the first $i-1$ rows and columns having only zeroes.  Now
the number of clows with head $i$ is given by the $[i,i]$ entry of $
A_i A_{i+1}^{n-2} A_i$.

\subsubsection*{The \emph{3D-matching} polynomial $\3dm = (\3dm_n)$:}
Consider the complete tripartite hyper-graph, where each partition
contain $n$ nodes, and each hyperedge has exactly one node from each
part.   As before, there are variables $X_e$ for hyperedge $e$ and
$Y_v$ for node $v$. 
\[\3dm_n := \sum_{M \subseteq A_n \times B_n \times C_n}
\left(\prod_{e \in M}  X_e^{q-1}\right)
\left(\prod_{\substack{v \in M\\ \textrm{(counted only once)}}} Y_v^{q-1} \right).\]

\noindent
(H'): Given an instance of 3D-Matching $\mathcal{H}$, we consider the
usual projection.  The variables corresponding to the vertices are all
set to $t$. The edges present in $\mathcal{H}$ are all set to $z$, and
the ones not present are set to $1$. Then the number of 3D-matchings
in $\mathcal{H}$, modulo $p$, is equal to the coefficient of
$z^{n(q-1)}t^{3n(q-1)}$ in $\3dm_n (z,t)$.

\noindent
(E'): To evaluate $\3dm_n$ over $\field_q$, consider the hypergraph
obtained after removing the vertices with zero weight, edges with zero
weight, and edges that contain a vertex with zero weight (even if the
edges themselves have non-zero weight). Every subset of hyperedges in
this modified hypergraph contributes $1$ to the total sum, and all
other subsets contribute 0. Hence, the evaluation equals
$2^{(\textrm{\# edges in the modified hypergraph})} \bmod p$.
\qed\end{proof}

It is worth noting that the cut enumerator polynomial $\cut$, showed
by B\"{u}rgisser to be $\vnp$-intermediate over field $\field_q$, is
in fact $\vnp$-complete over the rationals when $q=2$, \cite{dRA12}.
Thus the above technique is specific to finite fields.

\section{Monotone projection lower bounds}
\label{sec:monotone}
We now show that some of our intermediate polynomials are not
\emph{monotone} $p$-projections of the $\Perm$ polynomial.  The
results here are motivated by the recent results of
Grochow~\cite{Grochow15}.  Recall that a polynomial $f(x_1,\ldots
,x_n)$ is a \emph{projection} of a polynomial $g(y_1,\ldots,y_m)$ if
$f(x_1,\ldots ,x_n) = g(a_1,\ldots ,a_m)$, where $a_i$'s are either
constants or $x_j$ for some $j$. The polynomial $f$ is an
\emph{affine} projection of $g$ if $f$ can be obtained from $g$ by
replacing each $y_i$ with an affine linear function
$\ell_i(\tilde{x})$.  Over any subring of $\reals$, or more generally
any totally ordered semi-ring, a \emph{monotone projection} is a
projection in which all constants appearing in the projection are
non-negative.  We say that the family $(f_n)$ is a (monotone affine)
projection of the family $(g_n)$ with \emph{blow-up} $t(n)$ if for all
sufficiently large $n$, $f_n$ is a (monotone affine) projection of
$g_{t(n)}$.

\begin{theorem}
  \label{thm:mon-proj-lb}
Over the reals (or any totally ordered semi-ring), for any $q$, 
the families $\sat$ and $\ham$ are not monotone affine $p$-projections
of the $\Perm$ family.  Any monotone affine projection from $\Perm$ to
$\sat$ must have a blow-up of at least $2^{\Omega(\sqrt n)}$.  Any
monotone affine projection from $\Perm$ to $\ham$ must have a blow-up
of at least $2^{\Omega(n)}$.
\end{theorem}
Before giving the proof, we set up some notation. For more details,
see \cite{AT13,Rothvoss14,Grochow15}.  For any polynomial $p$ in $n$
variables, let $\newton(p)$ denote the polytope in $\reals^n$ that is
the convex hull of the vectors of exponents of monomials of $p$.  For
any Boolean formula $\phi$ on $n$ variables, let
\textsf{p-SAT}($\phi$) denote the polytope in $\reals^n$ that is the
convex hull of all satisfying assignments of $\phi$.  Let $K_n =
(V_n,E_n)$ denote the $n$-vertex complete graph. The travelling
salesperson (TSP) polytope is defined as the convex hull of the
characteristic vectors of all subsets of $E_n$  that  define a
Hamiltonian cycle in $K_n$.

For a polytope $P$, let $\efsize(P)$ denote the minimal number of
linear inequalities needed to define $P$. A polytope $Q \subseteq
\reals^m$ is an \emph{extension} of $P\subseteq \reals^n$ if there is
an affine linear map $\pi\colon \reals^m \to \reals^n$ such that
$\pi(Q) = P$. The \emph{extension complexity} of $P$, denoted
$\xc(P)$, is the minimum size $c(Q)$ of any extension $Q$ (of any dimension)
of $P$. 

The following are straightforward, see for instance \cite{Grochow15,FMPTW15}.
\begin{fact}
  \label{fact:xc}
  \begin{enumerate}
  \item $\efsize(\newton(\perm_n)) \leq 2n$.
  \item If polytope $Q$ is an extension of polytope $P$, then
    \(xc(P) \leq \xc(Q)\). 
  \end{enumerate}
\end{fact}

We use the following recent results.
\begin{proposition}
  \label{prop:xc}
  \begin{enumerate}
  \item Let $f(x_1,\ldots ,x_n)$ and $g(y_1,\ldots ,y_m)$ be
    polynomials over a totally ordered semi-ring $R$, with
    non-negative coefficients. If $f$ is a monotone projection of $g$,
    then the intersection of $\newton(g)$ with some linear subspace is
    an extension of $\newton(f)$.  In particular,
    $\xc(\newton(f)) \leq m + \efsize(\newton(g))$. \cite{Grochow15}
  \item For every $n$ there exists a 3SAT formula $\phi$ with $O(n)$
     variables and $O(n)$ clauses such that
     \(\xc(\mathsf{p\text{-}SAT}(\phi)) \geq 2^{\Omega(\sqrt{n})}.\)
          \cite{AT13}
   \item The extension complexity of the TSP 
     polytope is $2^{\Omega(n)}$. \cite{Rothvoss14}
  \end{enumerate}
\end{proposition}

\begin{proof}(of Theorem~\ref{thm:mon-proj-lb}.)
  Let $\phi$ be a 3SAT formula with $n$ variables and $m$ clauses
  as given by Proposition~\ref{prop:xc}~(2). 
For  the polytope
$P=\mathsf{p\text{-}SAT}(\phi)$,  $\xc(P)$ is high.

Let $Q$ be the Newton polytope of $\sat_n$. It resides in $N$
dimensions, where $N = n+ |\mathsf{Cl}_n|= n + 8n^3$, and is the
convex hull of vectors of the form $(q-1)\langle
\tilde{a}\tilde{b}\rangle$ where $\tilde{a} \in \{0,1\}^n$, $\tilde{b}
\in \{0,1\}^{N-n}$, and for all $c \in \mathsf{Cl}_n$, $\tilde{a}$
satisfies $c$ if and only if $b_c=1$. For each $\tilde{a} \in \{0,1\}^n$, there
is a unique $\tilde{b} \in \{0,1\}^{N-n}$ such that $(q-1)\langle
\tilde{a}\tilde{b}\rangle$ is in $Q$.

Define the polytope $R$, also in $N$ dimensions, to be the convex hull
of vectors that are vertices of $Q$ and also satisfy the constraint
$\sum_{c\in \phi} b_c \ge m$. This constraint discards vertices of $Q$
where $\tilde{a}$ does not satisfy $\phi$. Thus $R$ is an extension of
$P$ (projecting the first $n$ coordinates of points in $R$ gives a
$(q-1)$-scaled version of $P$), so by Fact~\ref{fact:xc}(2), $\xc(P)
\le \xc(R)$.  Further, we can obtain an extension of $R$ from any
extension of $Q$ by adding just one inequality; hence $\xc(R) \le 1 +
\xc(Q)$.

Suppose $\sat$ is a monotone affine projection of $\perm_n$ with
blow-up $t(n)$. By Fact~\ref{fact:xc}(1) and
Proposition~\ref{prop:xc}(1), $\xc(\newton(\sat)) =
\xc(Q)  \le t(n)
+ c(\perm_{t(n)}) \le O(t(n))$. From the preceding discussion and
by Proposition~\ref{prop:xc}(2), we get 
$2^{\Omega(\sqrt{n})} \le \xc(P) \le \xc(R) \le 1 + \xc(Q) \le
O(t(n))$. It follows that $t(n)$ is at least $2^{\Omega(\sqrt{n})}$.

For the $\ham$ polynomial, let $P$ be the TSP polytope and $Q$ be
$\newton(\ham)$. The vertices of $Q$ are of the form
$(q-1)\tilde{a}\tilde{b}$ where $\tilde{a} \in \{0,1\}^{{n \choose
    2}}$ picks a subset of edges, $\tilde{b} \in \{0,1\}^{n}$ picks a
subset of vertices, and the picked edges form a length-$n$ clow
touching exactly the picked vertices. Define polytope $R$ by
discarding vertices of $Q$ where $\sum_{i\in [n]}b_i < n$.  Now the
same argument as above works, using Proposition~\ref{prop:xc}(3) instead of
(4).
\qed\end{proof}

\section{Complete families for VP and VBP}
\label{sec:completeness}
The quest for a natural $\vp$-complete polynomial has generated a 
significant amount of 
research~\cite{Burgisser-book00,Raz10,Men11,CDM13,DMMRS14}. The first 
success story came from~\cite{DMMRS14}, where some naturally 
defined homomorphism polynomials were studied, and  a host of them
were  shown to be 
complete for the class $\vp$. But the results came with minor caveats. 
When the completeness was established under projections, there were non-trivial 
restrictions on the set of homomorphisms $\mathcal{H},$ and sometimes even 
on the target graph $H$. On the other hand, when  all 
homomorphisms were allowed, completeness could only be shown
under seemingly more 
powerful reductions, namely, constant-depth $c$-reductions. Furthermore, 
the graphs were either directed or had weights on nodes. It is worth
noting that the reductions in \cite{DMMRS14} actually do not 
use the full power of generic constant-depth $c$-reductions; a closer
analysis reveals that they are in fact  \emph{linear p-projection}. That is, 
the reductions are linear combinations of 
polynomially many $p$-projections (see Chapter 3,~\cite{Burgisser-book00}). 
Still, this falls short of $p$-projections. 

In this work, we  remove all such restrictions and show that 
there is a simple explicit homomorphism polynomial family that is complete for 
$\vp$ under $p$-projections. In this family, the source graphs $G$ are
specific bounded-tree-width graphs, and the target graphs $H$ are
complete graphs. We also show that a similar family with
bounded-path-width source graphs is complete for $\vbp$ under
$p$-projections. Thus, homomorphism polynomials are rich enough to
characterise computations by circuits as well as algebraic branching
programs. 

The polynomials we consider are defined formally as follows.
\begin{defi}
\label{def:hom}
Let $G = (V(G),E(G))$ and $H = (V(H),E(H))$ be two graphs. Consider the set of 
variables 
$\bar{Z} := \{Z_{u,a} \mid u \in V(G) \mbox{ and } a \in V(H)\}$ and 
$\bar{Y} := \{Y_{(u,v)}\mid (u,v) \in E(H)\}$. Let $\mathcal{H}$ be
a set of homomorphisms from $G$ to $H$. 
The  homomorphism polynomial $f_{G,H,\mathcal{H}}$ in the variable set
$\bar{Y}$, and the generalised homomorphism polynomial
$\hat{f}_{G,H,\mathcal{H}}$ in the variable set $\bar{Z} \cup
\bar{Y}$, are defined as follows:  
\begin{align*}
  f_{G,H,\mathcal{H}} & = \sum_{\phi \in \mathcal{H}}
  \left(\prod_{(u,v) \in  E(G)}Y_{(\phi(u),\phi(v))} \right). \\
  \hat{f}_{G,H,\mathcal{H}} & = \sum_{\phi \in \mathcal{H}}
  \left(\prod_{u \in V(G)} Z_{u,\phi(u)}\right)
  \left(\prod_{(u,v) \in  E(G)}Y_{(\phi(u),\phi(v))} \right). 
\end{align*}
Let $\hom$ denote the set of all homomorphisms from $G$ to $H$. 
If $\mathcal{H}$ equals $\hom$, then we drop it from the subscript and
write $f_{G,H}$ or $\hat{f}_{G,H}$. 
\end{defi}

Note that for every $G,H,\mathcal{H}$,
$f_{G,H,\mathcal{H}}(\bar{Y})$ equals
$\hat{f}_{G,H,\mathcal{H}}(\bar{Y}) \mid_{\bar{Z}=\bar{1}}$. 
Thus upper bounds for $\hat{f}$ give upper bounds for $f$, while lower
bounds for $f$ give lower bounds for $\hat{f}$. 

We show in Theorem~\ref{thm:ub} that for any $p$-family $(H_m)$, and any
bounded tree-width (path-width, respectively) $p$-family $(G_m)$, the
polynomial family $(f_m)$ where $f_m=\hat{f}_{G_m,H_m}$ is in $\vp$
($\vbp$, respectively). We then show in Theorem~\ref{thm:lb-vp} that
for a specific bounded tree-width family $(G_m)$, and for
$H_m=K_{m^6}$, the polynomial family $(f_{G_m,H_m})$ is hard, and
hence complete, for $\vp$ with respect to projections.  An analogous
statement is shown in Theorem~\ref{thm:lb-vbp} for a specific bounded
path-width family $(G_m)$ and for $H_m=K_{m^2}$. Over fields of
characteristic other than 2, $\vbp$-hardness is obtained for a simpler
family of source graphs $G_m$, as described in
Theorem~\ref{thm:lb-vbp-char-not-2}. 

\subsection{Upper Bound}
In~\cite{DMMRS14}, it was shown that the homomorphism polynomial 
$\hat{f}_{T_m,K_n}$ where $T_m$ is a binary tree on $m$ leaves, 
and $K_n$ is a complete graph on $n$ nodes, is 
computable by an arithmetic circuit of size $O(m^3n^3)$. 
Their proof idea is based on recursion: group the 
homomorphisms based on where they map the root of $T_m$ and its children, 
and  recursively compute the sub-polynomials within each group. The 
sub-polynomials of a specific group have a special set of variables 
in their monomials. 
Hence, the homomorphism polynomial can be computed by suitably combining 
partial derivatives of the sub-polynomials. 
The partial derivatives themselves can be computed efficiently using
the technique of Baur and Strassen, \cite{BS83}. 

Generalizing the above idea to polynomials where the source graph  is
not a binary tree $T_m$ but 
 a bounded tree-width graph $G_m$ seems hard. The very first obstacle we 
encounter is to generalize the concept of partial derivative to monomial 
extension. Combining sub-polynomials to obtain the original polynomial also 
gets rather complicated. 

We sidestep this difficulty by using a dynamic programming approach \cite{DST02}
based on a ``nice'' tree decomposition of the source graph. This shows that 
the homomorphism polynomial $\hat{f}_{G,H}$ is computable by an arithmetic circuit 
of size at most \(2|V(G)|\cdot|V(H)|^{tw(G)+1}\cdot (2|V(H)|+2|E(H)|),\) 
where $tw(G)$ is the tree-width of $G$. 

Let $\mathcal{T} = (T, \{X_t\}_{t\in V(T)})$ be a nice tree decomposition of $G$ 
of width $\tau$. 
 For each $t \in V(T)$, let $M_t = \{ \phi \mid \phi \colon X_t \to V(H)\}$ be 
the set of all mappings from $X_t$ to $V(H)$. Since $|X_t| \leq \tau+1$, we 
have 
$|M_t| \leq |V(H)|^{\tau+1}$. For each node $t \in V(T)$, let $T_t$ 
be the subtree 
of $T$ rooted at node $t$, $V_t := \bigcup_{t' \in V(T_t)}X_{t'}$, and 
$G_t := G[V_t]$ be the subgraph of $G$ induced on $V_t$. Note that $G_r = G.$

We will build the circuit inductively. For each $t \in V(T)$ and 
$\phi \in M_t$, we have a gate $\langle t, \phi \rangle$ in the circuit. Such 
 a gate will compute the homomorphism polynomial from $G_t$ to $H$ such 
that the mapping of $X_t$ in $H$ is given by $\phi$. For each such gate 
$\langle t, \phi \rangle$ we introduce another gate 
$\langle t, \phi \rangle'$ which computes the 
``partial derivative'' (or, quotient) of the polynomial computed at 
$\langle t, \phi\rangle$ 
with respect to the monomial given by $\phi$. As we mentioned before, the 
construction is inductive, starting at the leaf nodes and proceeding towards 
the root. 

\paragraph{Base case (Leaf nodes):} Let $\ell \in V(T)$ be a leaf node. 
Then, $X_\ell = \{u\}$ such that $u \in V(G)$. Note that any $\phi \in M_\ell$ is 
just a mapping of $u$ to some node in $V(H)$. Hence, the set $M_\ell$
can be identified with 
$V(H)$. Therefore, for all $h \in V(H)$, we label the gate 
$\langle \ell, h \rangle$ by the variable $Z_{u,h}$. The derivative gate 
$\langle \ell,h \rangle'$ in this case is set to $1$.

\paragraph{Introduce nodes:} Let $t \in V(T)$ be an introduce node, and $t'$ 
be its unique child. Then, $X_t \setminus X_{t'} = \{u\}$ for some
$u\in V(G)$. Let $N(u) := \{ v | v \in X_{t'} \mbox{ and }(v,u) \in E(G_t)\}$. 
Note that there is a one-to-one correspondence between $\phi \in M_t$ and pairs 
$(\phi', h) \in M_{t'} \times V(H)$. Therefore, for all 
$\phi (= (\phi',h)) \in M_{t}$ such that 
\(\forall v \in N(u), (\phi'(v),h) \in E(H),\)  we set 
\begin{align*}
\langle t, \phi\rangle & :=  Z_{u,h}\cdot\left(\prod_{v \in N(u)}Y_{(\phi'(v),h)}\right)\cdot \langle t', \phi' \rangle \;\;\;\;\;\mbox{ and}, \\ 
\langle t, \phi\rangle' & :=  \langle t', \phi'\rangle', 
\end{align*}
otherwise we set \(\langle t ,\phi\rangle = \langle t,\phi\rangle' := 0.\)

\paragraph{Forget nodes:} Let $t \in V(T)$ be a forget node and $t'$ be 
its unique child. Then, $X_{t'} \setminus X_{t} = \{u\}$ for some $u \in V(G)$. 
 Again note that there is a one-to-one correspondence between pairs 
\((\phi, h) \in M_t \times V(H)\) and $\phi' \in M_{t'}$. 
Let $N(u) := \{v| v \in X_{t}\mbox{ and }(v,u) \in E(G_{t'})\}.$
Therefore, for all $\phi \in M_t,$ we set 
\begin{align*}
\langle t,\phi\rangle & := \sum_{h \in V(H)} \langle t',(\phi,h)\rangle \;\;\;\;\;\mbox{ and}, \\ 
 \langle t,\phi \rangle' & :=  \sum_{\substack{h\in V(H)\mbox{ such that}\\ \forall v \in N(u), (\phi(v),h) \in E(H)}} Z_{u,h}\cdot\left(\prod_{v\in N(u)}Y_{(\phi(v),h)}\right) \cdot\langle t', (\phi,h)\rangle'. 
\end{align*}

\paragraph{Join nodes:} Let $t \in V(T)$ be a join node, and $t_1$ and $t_2$ be 
its two children; we have $X_t = X_{t_1} = X_{t_2}$. 
Then, for all $\phi \in M_t,$ we set 
\begin{align*}
\langle t,\phi\rangle & :=  \langle t_1, \phi\rangle\cdot\langle t_2,\phi\rangle' \left( =\langle t_1,\phi\rangle'\cdot \langle t_2,\phi\rangle\right) \\
\langle t,\phi\rangle' & :=  \langle t_1,\phi\rangle'\cdot\langle t_2,\phi \rangle'.
\end{align*}

The output gate of the circuit is $\langle r, \emptyset\rangle$. The 
correctness of the algorithm is readily seen via induction in a similar way. 
The bound on the size also follows easily from the construction. 

We observe some properties of our construction. First, the circuit 
constructed is a constant-free circuit. This was the case with the
algorithm from~\cite{DMMRS14} 
too. Second, if we start with a path decomposition, we obtain 
\emph{skew} circuits, since the \emph{join} nodes are absent. The
algorithm from \cite{DMMRS14} does not give skew circuits when $T_m$ is a path.
(It seems the obstacle there lies in computing partial-derivatives 
using skew circuits.) 

From the above algorithm and its properties, we obtain 
the following theorem.

\begin{theorem}
\label{thm:ub}
Consider the family of homomorphism polynomials $(f_m),$ 
where $f_m = f_{G_m,H_m}(\bar{Z},\bar{Y})$, and
$(H_m)$ is a $p$-family of complete graphs. 
 \begin{itemize}
\item If $(G_m)$ is a $p$-family of graphs of bounded tree-width, 
then $(f_m) \in \vp$.
\item If $(G_m)$ is a $p$-family of graphs of bounded path-width, 
then $(f_m) \in \vbp$.  
\end{itemize}
\end{theorem} 

\subsection{$\vp$-completeness}
We now turn our attention towards establishing $\vp$-\emph{hardness} 
of the homomorphism polynomials. We need to show that there exists 
a $p$-family $(G_m)$ of bounded tree-width graphs such that 
 $(f_{G_m,H_m}(\bar{Y}))$ is hard for $\vp$ under projections. 

We use \emph{rigid} and mutually \emph{incomparable} graphs in the 
construction of $G_m$. Let $I := \{I_0, I_1, I_2\}$ be a fixed set of three 
connected, rigid 
and mutually incomparable graphs. Note that they are necessarily 
\emph{non-bipartite}. Let $c_{I_i} = |V(I_i)|$. Choose an 
integer $c_{\max} > \max\;\{ c_{I_0},c_{I_1},c_{I_2} \}$. Identify two 
distinct vertices $\{v_{\ell}^0, v_{r}^0\}$ in $I_0$, three distinct vertices 
$\{v_\ell^1,v_r^1,v_p^1\}$ in $I_1$, and three distinct vertices 
$\{v_\ell^2,v_r^2,v_p^2\}$ in $I_2$. 

For every $m$ a power of 2, we denote a complete (perfect) binary tree
with $m$ leaves by $\tree_m$. We construct a sequence of graphs $G_m$
(Fig.~\ref{fig:rigid-tree}) from $\tree_m$ as follows: first replace
the root by the graph $I_0$, then all the nodes on a particular level
are replaced by either $I_1$ or $I_2$ alternately
(cf.~Fig.~\ref{fig:rigid-tree}). Now we add edges; suppose we are at a
`node' which is labeled $I_i$ and the left child and right child are
labeled $I_j$, we add an edge between $v^i_\ell$ and $v^j_p$ in the
left child, and an edge between $v^i_r$ and $v^j_p$ in the right
child. Finally, to obtain $G_m$ we expand each added edge into a simple path
with $c_{\max}$ vertices on it (cf.~Fig.~\ref{fig:rigid-tree}). That
is, a left-edge connection between two incomparable graphs in the tree
looks like, $I_i(v^i_\ell)\relbar\mbox{(path with }c_{\max}\mbox{
  vertices)}\relbar (v^j_p)I_j.$

\begin{figure}
\centering
\begin{tikzpicture}[vertex/.style={circle,draw,minimum size=1em,inner
      sep=3pt]}, scale=0.8]


\node[vertex] (1) at (0,0) {$I_0$};

\node[vertex] (2) at (-4,-1.5) {$I_1$};
\node[vertex] (3) at (4,-1.5) {$I_1$}; 

\node[vertex] (4) at (-6,-3) {$I_2$};
\node[vertex] (5) at (-2,-3) {$I_2$};
\node[vertex] (6) at (2,-3) {$I_2$};
\node[vertex] (7) at (6,-3) {$I_2$};

\node[vertex] (8) at (-7,-4.5) {$I_1$};
\node[vertex] (9) at (-5,-4.5) {$I_1$};
\node[vertex] (10) at (-3,-4.5) {$I_1$};
\node[vertex] (11) at (-1,-4.5) {$I_1$};
\node[vertex] (12) at (1,-4.5) {$I_1$};
\node[vertex] (13) at (3,-4.5) {$I_1$};
\node[vertex] (14) at (5,-4.5) {$I_1$};
\node[vertex] (15) at (7,-4.5) {$I_1$};

\draw[thick,dashed] (1) -- (2) -- (4) -- (8);
\draw[thick,dashed] (4) -- (9); 
\draw[thick,dashed] (2) -- (5) -- (10); 
\draw[thick,dashed] (5) -- (11);
\draw[thick,dashed] (1) -- (3)-- (6) -- (12);
\draw[thick,dashed] (6) -- (13); 
\draw[thick,dashed] (3) -- (7) -- (14);
\draw[thick,dashed] (7) -- (15);

\node (16) at (-6,-4.7) {};
\node (17) at (-6,-5.5) {}; 
\draw[thick,dotted] (16) -- (17);

\node (18) at (-2,-4.7) {};
\node (19) at (-2,-5.5) {}; 
\draw[thick,dotted] (18) -- (19);

\node (20) at (2,-4.7) {};
\node (21) at (2,-5.5) {}; 
\draw[thick,dotted] (20) -- (21);

\node (22) at (6,-4.7) {};
\node (23) at (6,-5.5) {}; 
\draw[thick,dotted] (22) -- (23);

\node (c1) at (-0.3,-0.2) {};
\node (c2) at (-3.5,-1.4) {};
\draw[decorate, decoration={brace,amplitude=4pt}] (c1) -- node[auto] {path with $c_{\max}$ vertices} (c2);

\end{tikzpicture}
\caption{The graph $G_m$.}\label{fig:rigid-tree}
\end{figure}
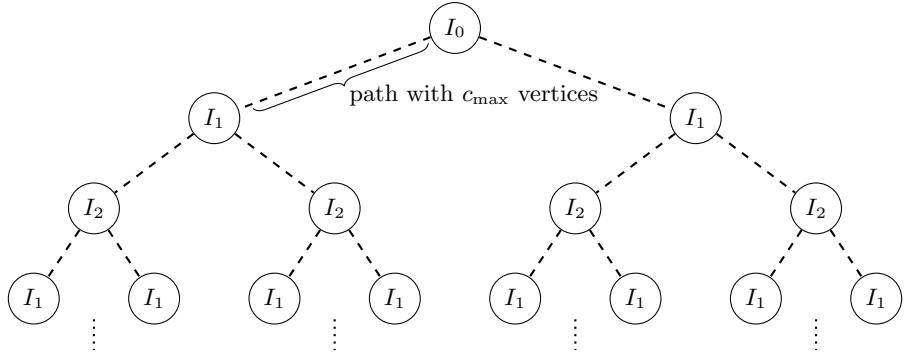

\begin{theorem}
\label{thm:lb-vp}
Over any field, the family of homomorphism polynomials 
$(f_m)$, with $f_m(\bar{Y}) = f_{G_m,H_m}(\bar{Y})$, where 
\begin{itemize}
\item $G_m$ is defined as above (see~Fig.~\ref{fig:rigid-tree}), and
\item $H_m$ is an undirected complete graph on $\mathsf{poly}(m)$, say $m^6$, 
vertices,
\end{itemize} 
is complete for $\vp$ under $p$-projections.
\end{theorem}
\begin{proof}
  \emph{Membership} in $\vp$   follows from Theorem~\ref{thm:ub}.

  We proceed with the \emph{hardness} proof. The idea is to obtain the
  $\vp$-complete universal polynomial from \cite{Raz10} as a
  projection of $f_m$. This universal polynomial is computed by a
  normal-form homogeneous circuit with alternating unbounded fanin-in
  $+$ and bounded fan-in $\times$ gates. We would like to put its
  parse trees in bijection with homomorphisms from $G$ to $H$. This
  becomes easier if we use an equivalent universal circuit in a
  nice normal form as described in \cite{DMMRS14}.
  The normal form circuit is {\em multiplicatively
    disjoint}; sub-circuits of $\times$ gates are disjoint (see
  \cite{MP08}).  This  ensures that even though $C_n$ itself is not a
  formula, all its   parse trees are already subgraphs of $C_n$ even
  without unwinding it into a formula.
  
  Our starting point is the related graph $J_n'$
  in~\cite{DMMRS14}. 
The parse trees in $C_n$ are complete alternating unary-binary
trees. The graph $J_n'$ is constructed  in such a way that the parse
trees are now in bijection with complete binary trees. To achieve
this, we ``shortcut'' the $+$ gates, while preserving information
about whether a subtree came in from the left or the right. 
For completeness sake we describe the construction of $J_n'$ from
\cite{DMMRS14}. 

We obtain a sequence of graphs 
$(J_n')$ from the undirected graphs underlying $(C_n)$ as follows.
 Retain the multiplication and input gates of $C_n$. Let us make two 
copies of each. For each retained gate, $g$, in $C_n$; let $g_L$ and
 $g_R$ be the two copies of $g$ in $J'_n$. We now define 
the edge connections in $J'_n$. Assume $g$
 is a $\times$ gate retained in $J'_n$. Let $\alpha$ and $\beta\,$ be 
two $+$ gates feeding into $g$ in $C_n$.
 Let $\{\alpha_1,\ldots,\alpha_i\}$ and $\{\beta_1,\ldots,\beta_j\}$ be 
the gates feeding into $\alpha$ and
 $\beta$, respectively. Assume without loss of generality that 
$\alpha$ and $\beta$ feed into $g$ from left
 and right, respectively. We add the following set of edges to $J'_n$:
 $\{(\alpha_{1L},g_L),\ldots,(\alpha_{iL},g_L)\}$, $\{(\beta_{1R},g_L),\ldots,(\beta_{jR},g_L)\}$, $\{(\alpha_{1L},g_R),\ldots,(\alpha_{iL},g_R)\}$
 and $ \{(\beta_{1R},g_R),\ldots,(\beta_{jR},g_R) \}$. We now would like to 
keep a single copy of $C_n$ in these set of edges. 
 So we remove the vertex $root_R$ and we remove the remaining spurious 
edges in following way. If we assume that all edges
 are directed from root towards leaves, then we keep only edges induced 
by the vertices reachable from $root_L$ in this directed graph. 
  In~\cite{DMMRS14}, it was observed that there is a one-to-one 
correspondence between parse trees of $C_n$ and subgraphs of $J_n'$ that 
are rooted at $root_L$ and isomorphic to $\tree_{2^{k(n)}}$. 

We now transform $J_n'$ using the set $I = \{I_0,I_1,I_2\}$. This is similar to 
the transformation we did to the balanced binary tree $\tree_m$. We 
replace each vertex  
by a graph in $I$; $root_L$ gets $I_0$ and the rest of the layers get 
$I_1$ or $I_2$ alternately (as in Fig.~\ref{fig:rigid-tree}).  Edge 
connections are made so that a left/right child is connected to its parent 
via the edge $(v^j_p,v^i_\ell)/(v^j_p,v^i_r)$. Finally we replace each edge 
connection by a path with $c_{\max}$ vertices on it 
(as in Fig.~\ref{fig:rigid-tree}), to obtain the graph $J_n$. All
edges of $J_n$  
are labeled 1, with the following exceptions:  Every input node contains the 
same rigid graph $I_i$. It has a vertex $v^i_p$. Each path connection 
to other nodes has this vertex as its end point. Label such path edges
that are incident on $v^i_p$ by the label of the input gate.  

Let  $m := 2^{k(n)}$. The choice of $\mathsf{poly}(m)$ is such that 
$4s_n \leq \mathsf{poly}(m)$, where $s_n$ is the size of $J_n$.
The $\bar{Y}$ variables are set to 
$\{0,1,\bar{x}\}$ such that the  non-zero variables pick out the 
graph $J_n$.  From the observations of \cite{DMMRS14} 
it follows 
that for each parse tree $p$-$\tree$ of $C_n$, there exists 
a homomorphism $\phi \colon G_{2^{k(n)}} \to J_n$ such that $\mon (\phi)$ 
is exactly equal to $\mon (p\mbox{-}\tree)$. By $\mon(\cdot)$ 
we mean the monomial associated with an object. 
We claim that these are the only valid homomorphisms 
from $G_{2^{k(n)}} \to J_n$.
We observe the following properties of homomorphisms from $G_{2^{k(n)}} \to J_n$, 
from which the claim follows. In the following by a rigid-node-subgraph 
we mean a graph in $\{I_0,I_1,I_2\}$ that replaces a vertex. 
\begin{enumerate}
\item[$(i)$] Any homomorphic image of a rigid-node-subgraph of 
$G_{2^{k(n)}}$ in $J_n$, cannot split across two mutually
  incomparable rigid-node-subgraphs in $J_n$.  
That is, there cannot be two vertices in a rigid subgraph of $G_{2^{k(n)}}$ 
such that one of them is mapped into a rigid subgraph say $n_1$, 
and the other one is mapped into another rigid subgraph say $n_2$. 
This follows because homomorphisms do not increase 
distance. 

\item[$(ii)$] Because of $(i)$, with each homomorphic image 
of a rigid node 
$g_i \in G_{2^{k(n)}}$, we can associate at most one rigid node of $J_n$, 
say $n_i$, such that the homomorphic image of $g_i$ is a subgraph of $n_i$
and the paths (corresponding to incident edges) emanating from it.
But such a subgraph has a homomorphism to $n_i$ itself: fold each 
hanging path into an edge and then map this edge into an edge within $n_i$. 
(For instance, let $\rho$ be a path hanging off $n_i$ and attached to
$n_i$ at $u$, and let $v$ be any neighbour of $u$ within $n_i$. Mapping
vertices of $\rho$ to $u$ and $v$ alternately preserves all edges and
hence is a homomorphism.)
Therefore, we note that in such a case we have a homomorphism 
from $g_i \to n_i$. By rigidity and mutual incomparability, $g_i$ must
be the same as $n_i$, and 
this folded-path homomorphism must be the identity map. 
The other scenario, where we cannot associate 
any $n_i$ because  $g_i$ is mapped entirely within connecting paths, is not
possible since it contradicts \emph{non-bipartiteness} of
mutually-incomparable graphs.   
\end{enumerate}
\textbf{Root must be mapped to the root:} The rigidity of $I_0$ and Property 
$(ii)$ implies that $I_0 \in G_{2^{k(n)}}$ is mapped identically to 
 $I_0$ in $J_n$. \\
\textbf{Every level must be mapped within the same level:} The 
children of $I_0$ in $G_{2^{k(n)}}$ are mapped to the children of the root 
while respecting left-right behaviour. Firstly, the left child cannot be 
mapped to the root because of incomparability of the graphs $I_1$ and $I_0$. 
Secondly, 
the left child cannot be mapped to the right child (or vice versa) 
even though they are the 
same graphs, because the minimum distance between the 
vertex in $I_0$ where the left 
path emanates and the right child is $c_{\max} + 1$ whereas the distance 
between the vertex in $I_0$ where the left path emanates
and the left child is $c_{\max}$. 
So some vertex from the left child must be mapped into the path leading to 
the right child and hence the rest of the left child must be mapped 
into a proper subgraph of right child. But this contradicts rigidity of $I_1$. 
Continuing like this, we can show that every level must 
map within the same level and that the mapping within a level is correct.  
\qed\end{proof}

\subsection{$\vbp$-completeness}
Finally, we  show that homomorphism polynomials are also rich enough to 
characterize computation by algebraic branching programs. 
Here we establish that there exists a $p$-family $(G_k)$ 
of undirected \emph{bounded path-width} graphs such that 
the family 
$(f_{G_k,H_k}(\bar{Y}))$ is $\vbp$-complete 
with respect to $p$-projections.

We note that for $\vbp$-completeness under projections, the
construction in~\cite{DMMRS14} 
required directed graphs.  In the undirected setting they 
could establish hardness only under \emph{linear p-projection}, that
too using $0$-$1$  valued weights.

As before, we use rigid and mutually incomparable graphs in the construction 
of $G_k$. Let $I := \{I_1, I_2\}$ be two connected, non-bipartite, 
rigid and mutually incomparable graphs. Arbitrarily pick vertices
$u \in V(I_1)$ and 
$v \in V(I_2)$. Let $c_{I_i} = |V(I_i)|$, 
and $c_{max} = \max\{c_{I_1},c_{I_2}\}$. 
 Consider the sequence of graphs 
G$_k$~(Fig.~\ref{fig:r8-path-r8}); for every $k$, 
there is a simple path with $(k-1)+2c_{max}$ edges between a copy of 
$I_1$ and $I_2$. The path is between the vertices $u \in V(I_1)$ and 
$v \in V(I_2)$. The path 
between vertices $a$ and $b$ in G$_k$ contains $(k-1)$ edges.   

\begin{figure}
\centering
\begin{tikzpicture}[vertex/.style={circle,draw,minimum size=1em,inner sep=3pt]}]

\node (8) at (0,0) {$I_1(u)$};

\node[vertex] (17) at (1.5,0) {};
\node[vertex] (18) at (2.5,0) {};
\node[vertex] (19) at (3.5,0) {$a$};
\node[vertex] (20) at (4.8,0) {};
\node[vertex] (21) at (6,0) {$b$};
\node[vertex] (22) at (7,0) {};
\node[vertex] (23) at (8,0) {};
\node (24) at (9.5,0) {$(v)I_2$};







\draw (8) -- (17);
\draw[dashed] (17) -- (18);
\draw (18) -- (19);
\draw[dashed] (19) -- (20) -- (21); 
\draw (21) -- (22);
\draw[dashed] (22) -- (23);
\draw (23) -- (24);

\node (c1) at (0.2,-0.5) {};
\node (c2) at (3.4,-0.5) {};
\draw[decorate, decoration={brace,mirror,amplitude=5pt}] (c1) -- (c2);
\node at (1.8,-1) {$c_{max}$ edges};

\node (c3) at (3.6,-0.5) {};
\node (c4) at (5.9,-0.5) {};
\draw[decorate, decoration={brace,mirror,amplitude=5pt}] (c3) -- (c4);
\node at (4.8,-1) {$k-1$ edges};

\node (c5) at (6.1,-0.5) {};
\node (c6) at (9.3,-0.5) {};
\draw[decorate, decoration={brace,mirror,amplitude=5pt}] (c5) -- (c6);
\node at (7.8,-1) {$c_{max}$ edges};


\end{tikzpicture}
\caption{The graph G$_k$.}\label{fig:r8-path-r8}
\end{figure}
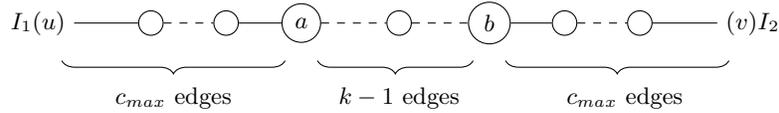

In other words, connect $I_1$ and $I_2$ by stringing together a path
with $c_{max}$ edges between $u$ and $a$, a path with $k-1$ edges
between $a$ and $b$, and a path with $c_{max}$ edges between $b$ and $v$.

\begin{theorem}
\label{thm:lb-vbp}
Over any field, the family of homomorphism polynomials 
$(f_k)$, where 
\begin{itemize}
\item $\textrm{G}_k$ is defined as above (see Fig.~\ref{fig:r8-path-r8}),
\item $H_k$ is the undirected complete graph on $O(k^2)$ vertices,
\item $f_k(\bar{Y}) = f_{\text{G}_k,H_k}(\bar{Y})$,
\end{itemize} 
is complete for $\vbp$ with respect to $p$-projections. 
\end{theorem}
\begin{proof}

\noindent\textbf{Membership:} It follows from Theorem~\ref{thm:ub}.

\noindent\textbf{Hardness:}  Let $(g_n) \in \vbp$. Without loss of
generality, we can assume that  
$g_n$ is computable by a 
layered branching program of polynomial size such that the number of layers, 
$\ell$, is more than the width of the algebraic branching program.

Let $B_n'$ be the undirected graph underlying the layered 
branching program $A_n$ for $g_n$. Let $B_n$ be the following graph: 
\(I_1(u) \relbar (s)B_n'(t) \relbar (v)I_2\), that is, 
 $u \in I_1$ is connected 
to $ s \in B_n'$ via a path with $c_{max}$ edges and  $ t \in B_n'$ is 
connected to $v \in I_2$ via a path with $c_{max}$ edges 
(cf. Fig.~\ref{fig:r8-path-r8}). 
The edges in $B_n'$ inherits the weight from $A_n$, and the rest of 
the edges in $B_n$ have weight $1$.

Let us now consider $f_\ell$ when the variables on the edges of $H_\ell$ 
are instantiated to values in $\{0,1\}$ or variables of $g_n$ so that we 
obtain $B_\ell$ as a subgraph of $H_\ell$. We claim that a valid homomorphism 
from G$_\ell \to B_\ell$ must satisfy the following properties: 
\begin{itemize}
\item[(P1)] $I_1$ in G$_\ell$ must be mapped to $I_1$ in $B_\ell$ using the 
identity homomorphism,

\item[(P2)] $I_2$ in G$_\ell$ must be mapped to $I_2$ in $B_\ell$ 
using the identity homomorphism.

\end{itemize}

Assuming the claim, it follows that homomorphisms from G$_\ell \to B_\ell$ are 
in one-to-one correspondence with $s$-$t$ paths in $A_n$. In particular, 
the vertex $a \in \text{G}_\ell$ is mapped to the 
vertex $s$ in $B_\ell$, and the vertex $b \in \text{G}_\ell$  
is mapped to the vertex $t$ in $B_\ell$. Also, the 
monomial associated with a homomorphism and its corresponding path are the 
same. Therefore, we have, \[ f_{\text{G}_\ell,B_\ell} = g_n .\]  Since $\ell$ is 
polynomially bounded, we obtain $\vbp$-completeness of $(f_k)$ over any field.

Let us now prove the claim. We first prove that a valid homomorphism from 
G$_\ell \to B_\ell$ must satisfy the property~(P1). There are three cases to 
consider.
\begin{itemize}
\item \textbf{Case 1:} \emph{Some vertex of $V(I_1) \subseteq V(G_\ell)$ is
  mapped to $u$ in $B_\ell$.}
Since homomorphisms cannot increase 
distances between two vertices, 
we conclude  that $V(I_1)$ must be mapped within the subgraph 
$I_1(u)-(a)$. Suppose further that some vertex on the $(u)-(a)$ path
other than $u$ is also in the homomorphic image of
$V(I_1)$. Some neighbour of $u$ in $V(I_1) \subseteq V(B_\ell)$, 
say $u'$,  must also be
in the homomorphic image, since otherwise we have a homomorphism from
the non-bipartite $I_1$ to a path, a contradiction.
But note that $I_1(u)-(a)$ has a homomorphism to $I_1$:  
fold the $(u)-(a)$ path onto the edge $u-u'$ in $I_1$. Hence, 
composing the two homomorphisms we obtain a homomorphism from $I_1$ to
$I_1$ which is not surjective. 
This  contradicts the rigidity of $I_1$. So in fact the homomorphism
must map $V(I_1)$ from $G_\ell$ entirely within $I_1$ from $B_\ell$,
and by rigidity of $I_1$, this must  be the identity map. 

\item \textbf{Case 2:} \emph{Some vertex of $V(I_1)\subseteq V(G_\ell)$ is
  mapped to $v$ in $B_\ell$.} 
Since homomorphisms cannot increase 
distances between two vertices,
we conclude  that $V(I_1)$ must be mapped within the subgraph 
$(b)-(v)I_2$. But note that $(b)-(v)I_2$ has a homomorphism to $I_2$ 
(fold the $(b)-(v)$ path onto any edge incident on $v$ within $I_2$). Hence, 
composing the two homomorphisms, we obtain a homomorphism from $I_1$ to $I_2$. 
This is a contradiction, since $I_1$ and $I_2$ were incomparable graphs to 
start with.

\item \textbf{Case 3:} \emph{No vertex of $V(I_1) \subseteq V(G_\ell)$
  is mapped to $u$  or $v$ in $B_\ell$.} Then  
$V(I_1)\subseteq V(\text{G}_\ell)$ must be mapped entirely within one of 
the following disjoint regions of $B_\ell$: $(a)$~$I_1 \setminus \{u\},$ $(b)$~bipartite 
graph between vertices  $u$  and  $v,$ and $(c)$~$I_2\setminus \{v\}$. 
But then we contradict \emph{rigidity of }$I_1$ in the first case, 
\emph{non-bipartiteness of }$I_1$ in the second case, 
and \emph{incomparability of }$I_1$ \emph{and} $I_2$ in the last.
\end{itemize}

In a similar way, we could also prove that a valid homomorphism from 
G$_\ell \to B_\ell$ must satisfy the property (P2). 
\qed\end{proof}


In the above proof, we crucially used incomparability of $I_1$ and
$I_2$ to rule out flipping an undirected path. It turns out that over
fields of characteristic not equal to 2, this is not crucial, since we
can divide by 2. 
We  show that if the characteristic of the underlying field is 
not equal to 2, then the sequence $(G_k)$ in the preceding theorem 
can be replaced by a sequence of simple undirected cycles of
appropriate length. 
In particular, we establish the following result. 

\begin{theorem}
\label{thm:lb-vbp-char-not-2}
Over fields of char $\neq2$, the family of homomorphism polynomials
 $(f_k)$, $f_k = f_{G_k,H_k},$ where 
\begin{itemize}
\item $G_k$ is a simple undirected cycle of length $2k+1$ and, 
\item $H_k$ is an undirected complete graph on $(2k+1)^2$ vertices, 
\end{itemize}
is complete for $\vbp$ under $p$-projections. 
\end{theorem}
\begin{proof}
\noindent\textbf{Membership:} As before, it follows from Theorem~\ref{thm:ub}. 
 
\noindent\textbf{Hardness:} Let $(g_n) \in \vbp$. Without loss of
generality, we can assume that  
$g_n$ is computable by a 
layered branching program of polynomial size satisfying the
 following properties:
\begin{itemize}
\item The number of layers, $\ell \geq 3$, is odd; say $\ell=2m+1$. So
  every path from $s$ to $t$ in the branching program has exactly $2m$
  edges.
\item The number of layers, is more than the width of the algebraic
  branching program,
\end{itemize} 

Let us consider $f_m$ when the variables on the edges of $H_m$
have been set to 0, 1, or variables of $g_n$ so that we 
obtain the undirected graph underlying the layered 
branching program $A_n$ for $g_n$ as a subgraph of $H_m$. Now change
the weight of the $(s,t)$ edge from 0 to weight $y$,
where $y$ is a new variable distinct from all the other variables of $g_n$. 
Call this modified graph $B_m$.  
Note that without the new edge, $B_m$ would be bipartite. 

Let us understand the homomorphisms from $G_m$ to $B_m$. Homomorphisms 
from a simple cycle $C$ to a graph $\mathcal{G}$ are in one-to-one correspondence 
with closed walks of the same length in $\mathcal{G}$. Moreover, if the 
cycle $C$ is of odd length, the closed walk must contain a simple odd cycle 
of at most the same length. Therefore, the only valid homomorphism 
from $G_m$ to $B_m$ 
are walks  of length $\ell=2m+1$, and they all contain the edge 
$(s,t)$ with weight $y$. But the cycles of length $\ell$ in $B_m$ are
in one-to-one  
correspondence with $s$-$t$ paths in $A_n$. Each cycle contributes 
$2\ell$ walks: we can start the walk at any of the $\ell$ vertices,
and we can follow the directions from $A_n$ or go against those
directions. Thus we have, 
\[f_{G_m,B_m} = (2 (2m+1)) \cdot y\cdot g_n 
= (2\ell) \cdot y\cdot g_n.\]

Let $p$ be the characteristic of the underlying field.
If $p=0$,  we substitute $y = (2\ell)^{-1}$ to obtain $g_n$. 
If $p > 2$, then 
$2\ell$ has an inverse if and only if $\ell$ has an 
inverse. Since $\ell \geq 3$ is an odd number, either $p$ does not divide 
$\ell$ or it does not divide $\ell+2$. Hence, at least one of $\ell$,
$\ell+2$ has an inverse. Thus 
$g_n$ is a projection of 
$f_m$ or $f_{m+1}$ depending on whether $\ell$ or $\ell+2$ has an inverse 
in characteristic $p$.

Since $\ell=2m+1$ is polynomially bounded in $n$, we therefore show $(f_k)$ is 
$\vbp$-complete with respect to $p$-projections over any field of 
characteristic not equal to 2. 
\qed\end{proof}


\section{Conclusion}
\label{sec:concl}
In this paper, we have shown that over finite fields, five families of
polynomials are intermediate in complexity between $\vp$ and $\vnp$,
assuming the PH does not collapse.  Over rationals and reals, we have
established that two of these families are provably not monotone
$p$-projections of the permanent polynomials.  Finally, we have
obtained a natural family of polynomials, defined via graph
homomorphisms, that is complete for $\vp$ with respect to projections;
this is the first family defined independent of circuits and with such
hardness. An analogous family is also shown to be complete for $\vbp$.

Several interesting questions  remain.

The definitions of our intermediate polynomials use the size $q$ of
the field $\field_q$, not just the characteristic $p$.  Can we find
families of polynomials with integer coefficients, that are
$\vnp$-intermediate (under some natural complexity assumption of
course) over all fields of characteristic $p$? Even more ambitiously,
can we find families of polynomials with integer coefficients, that
are $\vnp$-intermediate over all fields with non-zero characteristic?
at least over all finite fields? over fields $\field_p$ for all (or even for
infinitely many) primes $p$?

Equally interestingly, can we find an explicit family of polynomials
that is $\vnp$-intermediate in characteristic zero?

A related question is whether there are any polynomials defined over
the integers, that are $\vnp$-intermediate over $\field_q$ (for some
fixed $q$) but that are monotone $p$-projections of the permanent. 

Can we show that the remaining intermediate polynomials are also not
polynomial-sized monotone projections of the permanent? Do such
results have any interesting consequences, say, improved circuit lower
bounds?


\bibliographystyle{plain}
\bibliography{intermediate}

\end{document}